\documentclass[11pt, letterpaper]{article}
\usepackage{fullpage}
\usepackage{amsthm}
\usepackage{amsmath,amssymb,amsfonts,nicefrac}
\usepackage{xspace}
\usepackage{color}
\usepackage{url}
\usepackage{hyperref} 
\usepackage[noabbrev,capitalize,nameinlink]{cleveref}
\usepackage{bm}
\usepackage{bbm}
\usepackage{times}

\usepackage{enumitem}

\newtheorem{thm}{Theorem}[section]

\newtheorem{lemma}[thm]{Lemma}

\newtheorem{remark}[thm]{Remark}

\newtheorem{fact}[thm]{Fact}

\Crefname{lemma}{Lemma}{Lemmas}
\Crefname{fact}{Fact}{Facts}

\newcommand\E{\mathop{\mathbb{E}}}
\newcommand\card[1]{\left| {#1} \right|}
\newcommand\sett[2]{\left\{ \left. #1 \;\right\vert #2 \right\}}

\newcommand\Prob[2]{{\Pr_{#1}\left[ {#2} \right]}}

\newcommand\cProb[3]{{\Pr_{#1}\left[ \left. #3 \;\right\vert #2 \right]}}

\newcommand\eps{\varepsilon}

\renewcommand\geq{\geqslant}
\renewcommand\leq{\leqslant}

\newcommand{\rom}[1]{\uppercase\expandafter{\romannumeral #1\relax}}

\title{A Distance Amplification Lemma for Monotonicity}
\author{Dor Minzer\thanks{Department of Mathematics, Massachusetts Institute of Technology, Cambridge, USA. Supported by NSF CCF award 2227876 and 
 NSF CAREER award 2239160.}}
\date{\vspace{-5ex}}
\begin{document}
\maketitle
\begin{abstract}
    We show a procedure that, given oracle access to a function $f\colon \{0,1\}^n\to\{0,1\}$, produces oracle access to a function $f'\colon \{0,1\}^{n'}\to\{0,1\}$ such that 
    if $f$ is monotone, then $f'$ is monotone, and if $f$ is $\eps$-far from monotone, then $f'$ is $\Omega(1)$-far from monotone. Moreover, $n' \leq n 2^{O(1/\eps)}$ and each oracle query to $f'$ can be answered by making $2^{O(1/\eps)}$ oracle queries to $f$. 
    
    Our lemma is motivated by a recent result of [Chen, Chen, Cui, Pires, Stockwell, arXiv:2511.04558], who showed that for all $c>0$ there exists $\eps_c>0$, such that any (even two-sided, adaptive) algorithm distinguishing between monotone functions and $\eps_c$-far from monotone functions, requires $\Omega(n^{1/2-c})$ queries. Combining our lemma with their result implies a similar result, except that the distance from monotonicity is an absolute constant $\eps>0$, and the lower bound is $\Omega(n^{1/2-o(1)})$ queries.
\end{abstract}
\section{Introduction}
Consider the Boolean hypercube $\{0,1\}^n$ equipped with the partial ordering defined as 
$x\leq y$ if $x_i\leq y_i$ for all $i\in[n]$. A function $f\colon \{0,1\}^n\to\{0,1\}$ is called monotone if $x\leq y\Longrightarrow f(x)\leq f(y)$. The distance of $f$ from being monotone is defined as 
\[
\eps(f) = \min_{\substack{g\colon \{0,1\}^n\to\{0,1\}\\\text{monotone}}}\Prob{x\in\{0,1\}^n}{f(x)\neq g(x)}.
\]
In the monotonicity testing problem with proximity parameter $\eps$, one is given a function $f$ which is promised to either be monotone or else satisfy $\eps(f)\geq \eps$, and the task is to distinguish between these two cases. A one-sided $q$-query tester is a randomized algorithm that makes at most $q$ queries to the function $f$ and accepts with probability $1$ if $f$ is monotone and with probability at most $1/3$ if $f$ is $\eps$-far from being monotone. A two-sided $q$-query tester is an algorithm that makes at most $q$ queries to the function $f$ and accepts with probability at least $2/3$ if $f$ is monotone and with probability at most $1/3$ if $f$ is $\eps$-far from being monotone. We say a tester is non-adaptive if for each $i$, the $i$th query it makes is independent of the responses to the first $i-1$ queries. Otherwise, we say the tester is adaptive.

This problem, introduced by~\cite{GGLRS}, has received significant attention over the past two decades. The work~\cite{GGLRS} gave a one-sided, non-adaptive $O(n/\eps)$-query tester, and this result was subsequently improved by a sequence of works~\cite{CS,CST,KMS}, ultimately leading up to a one-sided, non-adaptive $\tilde{O}(\sqrt{n}/\eps^2)$-query tester. 
On the lower-bounds front, the work~\cite{FLNRRS} showed that any one-sided, non-adaptive tester must make $\Omega(\sqrt{n})$-queries, and this was extended to all two-sided, non-adaptive testers in~\cite{CST,CDST,CWX}. These lower bounds did not address adaptive testers, and in particular they did not rule out a $\Theta_{\eps}(\log n)$-query adaptive algorithm for monotonicity.

This changed with the work~\cite{BB}, who established a lower bound of $\tilde{\Omega}(n^{1/4})$ for two-sided, adaptive algorithms. This lower bound was subsequently improved to $\tilde{\Omega}(n^{1/3})$ by~\cite{CWX}. Very recently, a nearly tight lower bound for two-sided, adaptive testers was proved in~\cite{CCCPS}, as stated below.
\begin{thm}\label{thm:CCCPS}
    For all $c>0$, there exists $\eps_c>0$ such that any two-sided, adaptive tester for testing whether a given Boolean function $f\colon\{0,1\}^n\to\{0,1\}$ is monotone or $\eps_c$-far from monotone, requires $\Omega(n^{1/2-c})$ queries. Asymptotically, $\eps_c = 2^{-\Theta(1/c)}$.
\end{thm}
\Cref{thm:CCCPS} gives a nearly optimal lower bound for monotonicity testing, and the only slight downside of it is that the parameter $\eps_c$ is vanishing as $c$ approaches $0$. In particular, it does not rule out that for some small constant $\eps>0$, say $\eps = 0.001$, there is an $O(n^{0.49})$ query tester for monotonicity. 
This motivates the question of whether there is a procedure that amplifies distance from monotonicity (while keeping the function monotone if it was so in the first place). In this note we show an affirmative answer to this question, formalized in the following lemma.
\begin{lemma}\label[lemma]{lem:main}
   There exists an absolute constant $\delta>0$ such that for all $\eps>0$, there is a randomized transformation that makes $2^{\Theta(1/\eps)}$ oracle queries to an $n$-variate function $f\colon\{0,1\}^n\to\{0,1\}$, and produces an oracle access to an $n'$-variate function $f'\colon\{0,1\}^{n'}\to\{0,1\}$, with the following properties:
   \begin{enumerate}
    \item We have that $n' = n\cdot 2^{\Theta(1/\eps)}$ and a single query to $f'$ can be simulated by $2^{\Theta(1/\eps)}$ queries to $f$.
    \item If $f$ is monotone, then with probability $1$ we have that $f'$ is monotone.
    \item If $\eps(f)\geq \eps$, then with probability at least $0.99$ we have $\eps(f')\geq \delta$.
\end{enumerate}
\end{lemma}
Using~\Cref{lem:main}, we immediately deduce from~\Cref{thm:CCCPS} the following result:
\begin{thm}\label{thm:main}
    There exists $\delta>0$ such that any two-sided, adaptive tester for testing whether a given Boolean function $f\colon\{0,1\}^n\to\{0,1\}$ is monotone or $\delta$-far from monotone, requires $\Omega(n^{1/2-o(1)})$ queries.
\end{thm}
\begin{remark}
The argument we give for~\Cref{lem:main} (and subsequently for~\Cref{thm:main}) can be adapted to give a $\delta$ arbitrarily close to $\frac{\ln 2}{2}$. It may be possible to achieve $\delta$ arbitrarily close to $1/2$, and also to improve the blow-up to be polynomial in $1/\eps$. Getting $\delta$ close to $1/2$ would be optimal (as any Boolean function is $1/2$-close to some constant function, which is monotone). Improving the blow-up to polynomial in $1/\eps$ would improve $o(1)$ in the exponent in~\Cref{thm:main}, but would not directly lead to a lower bound of $\tilde{\Omega}(\sqrt{n})$ for adaptive, two-sided testers. For that, it seems necessary to improve~\Cref{thm:CCCPS} to have $\eps_c = c^{\Theta(1)}$.
\end{remark}
\section{Preliminaries}
\subsection{Violating Pairs, Matchings and Distance from Monotonicity}
Fix a function $f\colon \{0,1\}^n\to\{0,1\}$. A pair of points $x,y\in\{0,1\}^n$ is said to be a violating pair if $x<y$ and $f(x)>f(y)$ or $x>y$ and $f(x)<f(y)$. The violation graph of $f$, denoted by $G_{{\sf viol}}^f$, is a bipartite graph with sides $f^{-1}(1)$ and $f^{-1}(0)$, and whose edges are the violating pairs. The following result from~\cite{FLNRRS} gives a connection between $\eps(f)$ and the size of the largest matching in 
$G_{{\sf viol}}^f$.
\begin{lemma}\label[lemma]{lem:mono_vs_matching}
    For any function $f\colon \{0,1\}^n\to\{0,1\}$ we have that $\eps(f) 2^{n}$ is equal to the size of the largest matching in $G_{{\sf viol}}^f$.
\end{lemma}
\begin{proof}
    By~\cite[Corollary 2]{FLNRRS} we have that $\eps(f)2^n$ is equal to the size of the smallest vertex cover in $G_{\sf viol}^f$. The result now follows by K\H{o}nig's theorem.
\end{proof}

\subsection{The Tribes Function}\label{sec:tribes}
We will use the tribes function from~\cite{BenOrLinial}. Typically, it is studied when the hypercube $\{0,1\}^k$ is 
equipped with the uniform measure. For technical reasons we will need to consider 
the hypercube with the $p$-biased measure $\mu_p(x) = p^{|x|}(1-p)^{k-|x|}$ for $p$ very close to some $\alpha\in (0,1)$ bounded away from $0$ and $1$. 

For an integer 
$k\in\mathbb{N}$ we take $T_1,\ldots,T_r\subseteq [k]$ a maximal collection of disjoint sets, all of size $\log_{1/\alpha} k-10\log_{1/\alpha}\log k$. We note that 
$r\geq \frac{k}{\log_{1/\alpha} k}$ and
\begin{equation}\label{eq:calc_alpha}
\Prob{x\sim \mu_{\alpha}}{\bigwedge_{j\in T_1} x_j=1} = \alpha^{\log_{1/\alpha} k-10\log_{1/\alpha}\log k}
=\frac{\log^{10} k}{k},
\end{equation}
so we may choose $r'\leq r$, say $r' = \lfloor\frac{k\ln 2}{\log^{10} k}\rfloor$, such that defining the function $T\colon \{0,1\}^k\to\{0,1\}$ by
\[
T(x) = \bigvee_{i=1}^{r'}\bigwedge_{j\in T_i} x_j,
\]
we have that $\Prob{x\sim\mu_{\alpha}}{T(x)=1}= \frac{1}{2}+o(1)$. We fix such a choice and state a few well-known facts about the tribes function. First, we show the probability $T$ is $1$ is roughly $1/2$ also under $\mu_p$.
\begin{fact}\label[fact]{fact:compute_avg_tribe}
    Suppose that $\frac{1}{10}\leq \alpha\leq \frac{9}{10}$, 
    $\card{p-\alpha}\leq \frac{1}{k}$ and that $k$ is large enough. Then 
    \[
    \Prob{x\sim\mu_{p}}{T(x)=1}
    =\frac{1}{2}+o(1).
    \]
\end{fact}
\begin{proof}
  Note that sampling $x\sim\mu_p$, the probability that $T_1(x) = 1$ is equal to
   \begin{equation}\label{eq:p_power_est}
   p^{\log_{1/\alpha} k-10\log_{1/\alpha}\log k}
   =\frac{\log^{10}k}{k}\left(1+\frac{p-\alpha}{\alpha}\right)^{\log_{1/\alpha} k-10\log_{1/\alpha}\log k}
   =\frac{\log^{10}k}{k} \left(1+\Theta\left(\frac{\log k}{k}\right)\right),
   \end{equation}
   where in the first transition we used~\eqref{eq:calc_alpha}. It follows that
   \begin{align*}
   \Prob{x\sim \mu_{p}}{T(x)=1}
   &=1-\left(1-\frac{\log^{10}k}{k} \left(1+\Theta\left(\frac{\log k}{k}\right)\right)\right)^{r'}\\
   &=1-\Prob{x\sim \mu_{\alpha}}{T(x)=0}
   \left(\frac{1-\frac{\log^{10} k}{k}\pm \Theta\left(\frac{\log^{11} k}{k^2}\right)}{1-\frac{\log^{10} k}{k}}\right)^{r'}\\
   &=1-\Prob{x\sim \mu_{\alpha}}{T(x)=0}
   \left(1\pm \Theta\left(\frac{\log^{11} k}{k^2}\right)\right)^{r'},
   \end{align*}
   which is equal to $\frac{1}{2}+o(1)$.
\end{proof}

We next need a few facts about the vertex-boundary of the function $T$. For a given input $x\in \{0,1\}^k$, we say a tribe $T_i$ is almost satisfied if all but one its variables are assigned $1$ in $x$. We denote by ${\sf as}_T(x)$ the number of tribes $T_i$ that are almost satisfied. 
\begin{fact}\label[fact]{fact:almost_sat}
Suppose that $\frac{1}{10}\leq \alpha\leq \frac{9}{10}$, 
    $\card{p-\alpha}\leq \frac{1}{k}$ and that $k$ is large enough. Then for all $\delta>0$
    \[
    \Prob{x\sim\mu_p}{{\sf as}_T(x)\geq (\ln 2+\delta)\log_{1/\alpha} k}\leq k^{-\Omega_{\delta}(1)}.
    \]
\end{fact}
\begin{proof}
    Sampling $x$, we note that the events $E_i$ that $T_i$ is nearly satisfied are independent. The probability that a fixed tribe $T_i$ is almost satisfied is equal to
    \[
    (\log_{1/\alpha} k-10\log_{1/\alpha}\log k)
    \cdot p^{\log_{1/\alpha} k-10\log_{1/\alpha}\log k-1}
    \leq (1+o(1))\log_{1/\alpha} k\frac{\log^{10} k}{k}, 
    \]
    where we used~\eqref{eq:p_power_est}. Thus, letting $Z$ be the number of tribes that are almost satisfied, we have that $\E[Z] \leq r'\cdot (1+o(1))\log_{1/\alpha} k\frac{\log^{10} k}{k}  \leq (\ln 2+o(1))\log_{1/\alpha} k$, and the result follows by Chernoff's bound.
\end{proof}

For an input $x\in \{0,1\}^k$ such that $T(x)=1$, we say that $x$ is on the vertex-boundary if there is a unique tribe $T_i$ that evaluates to $1$ under $x$.
\begin{fact}\label[fact]{fact:vtx_boundary}
Suppose that $\frac{1}{10}\leq \alpha\leq \frac{9}{10}$, 
    $\card{p-\alpha}\leq \frac{1}{k}$ and that $k$ is large enough. Then
    \[
    \Prob{x\sim\mu_p}{T(x) = 1\text{ and $x$ is on the vertex-boundary}} =  \frac{\ln 2}{2}(1+o(1)).
    \]
\end{fact}
\begin{proof}
    Fix a tribe $T_i$. The probability that it is the only tribe evaluating to $1$ under a randomly chosen $x\sim \mu_p$ is equal to
    \begin{align*}
    &p^{\log_{1/\alpha} k-10\log_{1/\alpha}\log k}
    (1-p^{\log_{1/\alpha} k-10\log_{1/\alpha}\log k})^{r'-1}\\
    &=\frac{\log^{10}k}{k} \left(1+\Theta\left(\frac{\log k}{k}\right)\right)
    \left(1-\frac{\log^{10}k}{k} \left(1+\Theta\left(\frac{\log k}{k}\right)\right)\right)^{r'-1}\\
    &=\frac{\log^{10}k}{k} \left(\frac{1}{2}+o(1)\right),
    \end{align*}
    where in the second transition we used~\eqref{eq:p_power_est}, and in the last transition we used the choice of $r'$. Since these events are disjoint, the probability in question is equal to
    \[
    r'\frac{\log^{10}k}{k} \left(\frac{1}{2}+o(1)\right) = \frac{\ln 2}{2}(1+o(1)).\qedhere
    \]
\end{proof}
\section{Proofs}
\subsection{Proof of~\Cref{lem:main}}
Let $f$ be a function as in the statement of the lemma, and let $p$ denote its average. First, we argue that we may assume without loss of generality that $\frac{1}{3}\leq p\leq \frac{2}{3}$. Indeed, we define $F\colon \{0,1\}^{n+6}\to\{0,1\}$ 
as $F(y,x) = f(x)$ if $|y|=3$, $F(y,x) = 0$ if $|y|\leq 2$ and $F(y,x) = 1$ if $|y|\geq 4$. It is easy to see that the probability $F$ gets each one of the values $0$ and $1$ is at least $\frac{15+6+1}{64}\geq \frac{1}{3}$. Also, if $f$ is monotone then $F$ is monotone, and if $\eps(f)\geq \eps$, then $\eps(F)\geq \Omega(\eps)$. 

We now describe the randomized transformation.
Let $k = 2^{\frac{C}{\eps}}$ be an integer for $C>0$ a sufficiently large absolute constant. We query $f$ at $k^3$ uniformly and independently chosen points, and let $\alpha$ be the resulting empirical average. Take $T\colon \{0,1\}^k\to\{0,1\}$ to be the tribes from~\Cref{sec:tribes} corresponding to this $\alpha$, and define $f'\colon \{0,1\}^{k n}\to\{0,1\}$ as follows. First, interpret the input $x$ as $k$ blocks $x(1),\ldots,x(k)\in \{0,1\}^n$, and set
\[
f'(x) = T(f(x(1)),\ldots,f(x(k))).
\]
This completes the description of the transformation, and we analyze its properties. For the first item, to query $f'$ on input $x$, we need to query $f$ at the $k$ inputs $x(1),\ldots,x(k)$. For the second item, if $f$ is monotone, then $f'$ is the composition of monotone functions and hence it is also monotone. The rest of the argument is devoted for the third item. Note that with probability at least $0.99$, we have that $\card{\alpha -p}\leq \frac{1}{k}$. We show that in that case, $f'$ is $\Omega(1)$-far from monotone.

Using~\Cref{lem:mono_vs_matching}, we conclude that the bipartite graph $G_{{\sf viol}}^{f}$ has a matching of size $\eps 2^n$. Take $X\subseteq \{0,1\}^n$ to be the left endpoints of this matching, i.e., the side of inputs evaluating to $1$ under $f$. For each $x\in X$ let $M(x)$ be the point it is matched to.

Sample $x\in\{0,1\}^{k n}$, and let $y = (f(x(1)),\ldots,f(x(k)))\in \{0,1\}^k$. Note that the distribution of $y$ is $\mu_p$. Combining~\Cref{fact:vtx_boundary,fact:almost_sat} we conclude that with probability at least $\frac{\ln 2}{2}-0.01$ over the choice of $x$, we have that ${\sf as}_T(y)\leq (\ln 2 + 0.01)\log_{1/\alpha} k$, $T(y) = 1$ and $y$ is on the vertex-boundary of $T$; we denote this event by $E$ and condition on it henceforth. Thus, there is a unique tribe $T_i$ which is fully satisfied, and we write $T_i = \{j_1,\ldots,j_s\}$. We note that the points $x(j_1),\ldots,x(j_s)$ are distributed independently and uniformly in $f^{-1}(1)$, so each one of them falls in $X$ with probability $\eps/p$. Letting $Z(x)$ be the number of $\ell$ such that $x(j_{\ell})\in X$, we have that $\E[Z] = \eps s/p$, and
as $s = \log_{1/\alpha}k - 10\log_{1/\alpha}\log k = \Theta(C/\eps)$ Chernoff's bound gives that
\[
\cProb{}{E}{\card{Z(x)-\eps s/p} \leq 0.01\eps s/p}\geq 0.99
\]
for large enough $C$. Define 
\[
W = \sett{x}{x\in E, \card{Z(x)-\eps s/p}\leq 0.01\eps s/p},
\]
then $\Prob{x\in\{0,1\}^{kn}}{x\in W}\geq \frac{\ln 2}{2}-0.02$.

We construct a bipartite graph $H$ whose left side is $W$, and whose right side is ${f'}^{-1}(0)$. For each $x\in W$ let $T_i$ be the unique tribe evaluating to $1$, and let $T_i'\subseteq T_i$ be the subset of coordinates $j\in T_i$ such that $x(j)\in X$. For each $j\in T_i'$ we add an edge to $H$ between $x$ and the point $x'$, which differs from $x$ only on block $j$, on which it is equal to $M(x(j))$.

Clearly, each $x\in W$ has degree $Z(x)\geq 0.99\eps s/p$, so the number of edges in $H$ is at least $0.99\eps s|W|/p$. We now estimate the right degrees in $W$. Sample $x\in {f'}^{-1}(0)$ uniformly and set $y = (f(x(1)),\ldots,f(x(k)))$. If ${\sf as}_T(y)\geq (\ln 2 + 0.01)\cdot \log_{1/\alpha} k + 2$, then it has degree $0$, so we condition henceforth on the event that ${\sf as}_T(y)\leq (\ln 2+0.01)\cdot \log_{1/\alpha} k +1\leq (\ln 2+0.02)\cdot \log_{1/\alpha} k$. 
Let $I = \{i~|~ T_i \text{ is almost satisfied by $x$}\}$, so that $|I| =  {\sf as}_T(y)\leq  (\ln 2+0.02)\cdot \log_{1/\alpha} k$.
Then for each $i\in I$ we take $j_i\in T_i$ to be the unique block such that $f(x(j_i))=0$. Then 
the degree of $x$ can be written as
\[
d(x) = \sum\limits_{i\in I}1_{x(j_i)\text{ matched in $M$}}.
\]
Condition on $I$ and on $\{j_i\}_{i\in I}$. 
Note that the distribution of $x(j_i)$ for $i\in I$ is independent and each one is uniformly distributed in $f^{-1}(0)$, so we conclude that $\E[d(x)] = |I|\frac{\eps}{1-p}\leq (\ln 2+0.02)\frac{\eps}{1-p}\log_{1/\alpha} k
\leq (\ln 2+0.02+o(1))\frac{\eps}{1-p} s$. It follows from Chernoff's bound that for all $L\geq 1$,
\[
\Prob{x}{d(x)\geq L\cdot (\ln 2+0.03)\frac{\eps}{1-p} s}
\leq e^{-\Omega(L\log_{1/\alpha} k)} = k^{-\Omega(L)}.
\]
Thus, removing from $H$ all vertices $x\in {f'}^{-1}(0)$ with $d(x)\geq (\ln 2+0.03)\frac{\eps}{1-p} s$, we get that the number of edges removed from $H$ is at most
\[
\sum\limits_{L=1}^{\infty}  2L\cdot (\ln 2+0.03)\frac{\eps}{1-p} s \cdot k^{-\Omega(L)} 2^{nk} 
\leq 
2\frac{\eps}{1-p} s \cdot k^{-\Omega(1)}
2^{nk}
\leq
2^{nk},
\]
where the last inequality holds for sufficiently large $C$. Let $H'$ be the graph $H$ after removal of these vertices.

Note that $H'$ is a bipartite graph, with sides of size at most $2^{nk}$, and the number of edges in it is 
\[
0.99\frac{\eps}{p} s|W| - 2^{nk}
\geq 
\left(\frac{\eps}{p} s \left(\frac{\ln 2}{2}-0.03\right) -1) \right)2^{nk}
\geq  
\left(\frac{\ln 2}{2}-0.04\right) \frac{\eps}{p} s \cdot 2^{nk}.
\]
where we used the fact $C$ is sufficiently large. Also, the degree of each left vertex is at most $1.01\frac{\eps}{p} s$, and the degree of each right vertex is at most $(\ln 2+0.03)\frac{\eps}{1-p} s$. Thus, running the greedy algorithm shows that $H'$ has a matching of size at least $\frac{\left(\frac{\ln 2}{2}-0.03\right)\frac{\eps}{p} s \cdot 2^{nk}}{1.01\eps\max(1/p,1/(1-p)) s}\geq \Omega(2^{nk})$, and we note that each edge of this matching is a violating pair for $f'$. It follows from~\Cref{lem:mono_vs_matching} that 
$\eps(f')\geq \Omega(1)$.\hfill\qedsymbol
\subsection{Proof of~\Cref{thm:main}}
Let $n\in\mathbb{N}$ be large and set $c=\frac{\xi}{\log\log n}$ for sufficiently small absolute constant $\xi>0$. Using~\Cref{thm:CCCPS} we get that 
there is a distribution $\mathcal{D}_{{\sf yes}}$ over $n$-variate monotone functions and a distribution $\mathcal{D}_{{\sf no}}$ over $n$-variate, $\eps_c$-far from monotone functions, such that for any $O(n^{1/2-c})$-query algorithm ${\sf ALG}$,
\[
\card{\Prob{f\sim \mathcal{D}_{{\sf yes}}}{{\sf Alg}(f)}
-\Prob{f\sim \mathcal{D}_{{\sf no}}}{{\sf Alg}(f)}}
\leq 0.1.
\]
Furthermore, $\eps_c = 2^{-\Theta(1/c)} = (\log n)^{-\Theta(\xi)}\geq \frac{1}{\sqrt{\log n}}$, where in the last inequality we used that $\xi$ is sufficiently small.

Take $\delta>0$ from~\Cref{lem:main}, and suppose that there is a $q(n')$-query algorithm ${\sf Alg}'$ that distinguishes between $n'$-variate monotone functions and $n'$-variate $\delta$-far from monotone functions. We now describe an algorithm ${\sf Alg}$ distinguishing monotone functions and $\eps_c$-far from monotone functions. Given oracle access to an $n$-variate function $f$, set $k = 2^{\Theta(1/\eps_c)}$ and use~\Cref{lem:main} to construct an oracle access to a function $f'$. Run the algorithm ${\sf Alg}'$ on $f'$, and accept only if it does.

First, note that the query complexity of ${\sf Alg}'$ is $k \cdot q(k n)$. Second, note that if $f$ is monotone, then $f'$ is monotone, so by the premise of ${\sf Alg}'$ we get that the algorithm accepts with probability at least $2/3$. Third, note that if $f$ is $\eps_c$-far from monotone, then by~\Cref{lem:main} with probability at least $0.99$ we have that $f'$ is $\delta$-far from monotone, and then ${\sf Alg}'$ accepts with probability at most $1/3$. Thus, if $f$ is $\eps_c$-far from monotone, ${\sf Alg}$ accepts with probability at most $1/3 + 0.01$. Overall, we get that 
\[
\card{\Prob{f\sim \mathcal{D}_{{\sf yes}}}{{\sf Alg}(f)}
-\Prob{f\sim \mathcal{D}_{{\sf no}}}{{\sf Alg}(f)}}
\geq \frac{1}{3} - 0.01 > 0.1
\]
By choice of $\mathcal{D}_{{\sf yes}}$ and $\mathcal{D}_{{\sf no}}$ this implies that $kq(kn)\geq \Omega(n^{1/2-c})$. 
Setting $n' = kn = n^{1+o(1)}$ we get that 
$q(n')\geq n'^{1/2-o(1)}$, concluding the proof.
\hfill\qed

\bibliographystyle{plain}
\bibliography{ref}
\end{document}